\documentclass[12pt,a4paper]{article}
\usepackage{amsmath,amsfonts,amssymb,amsthm, epsfig,
epstopdf, titling, url, array}
\usepackage[utf8]{inputenc}
\usepackage{mathtools}
\usepackage[english]{babel}
\usepackage{hyperref}
\newtheorem{thm}{Theorem}[section]
\newtheorem{lem}{Lemma}[section]
\newtheorem{cor}{Corollary}[section]
\newtheorem{prop}{Proposition}[section]
\numberwithin{equation}{section}

\def\C{\mathbb C}
\def\D{\mathbb D}

\def\R{I\!\!R}
\def\H{I\!\!H}
\def\C{I\!\!\!\!C}

\def\D{I\!\!D}
\title{ WAVE KERNELS WITH
MAGNETIC FIELD ON THE HYPERBOLIC PLANE AND WITH THE MORSE POTENTIAL ON THE REAL LINE}
\author{Mohamed Vall Ould Moustapha}
\begin{document}
\maketitle
\begin{abstract}
In this article we give explicit solutions for the wave equations associated to the modified Schr\"odinger operators with magnetic field on the disc and the upper half plane models of the hyperbolic plane.  We show that the modified  Schr\"odinger operator with magnetic field on the upper half plane model and the Schr\"odinger operator with diatomic molecular Morse potential on $\R$ are related by means of one-dimensional Fourier transform. Using this relation we give the explicit forms of the wave kernels associated to the Schr\"odinger operator with the diatomic molecular  Morse potential on $\R$ in terms of the two variables confluent hypergeometric function $\Phi_1$. 
\end{abstract}
 
\section{Introduction}
 
The aim of this
paper is to give explicit solutions to the wave equations associated to the modified Schr\"odinger operators with magnetic field on the disc and the upper half plane models of the hyperbolic plane. Using a formula relating the Schr\"odinger operator with with magnetic field on the hyperbolic plane and the Schr\"odinger operator with Morse potential  on the real line we give an exact formula for wave kernel with the Morse potential. That is we obtain the explicit solutions for the following
Cauchy problems of the wave type:
\begin{align}\label{cauchy-D} \left \{\begin{array}{cc}{\cal D}_{k} u(t, w)=\frac{\partial^2}{\partial
t^2}u(t, w)
, (t, w)\in \R^\ast_+\times \D \\ u(0, w)=0, u_t(0, w)=u_1(w), u_1\in
C^\infty_0(\D)\end{array}
\right. \end{align}
%%%%%%%%%%%%%%%%%%%%%%%%%%%%%%%%%%%%%%%%%%%%%%%%%%%%%%%%%%%%%%%%%%%%%

%%%%%%%%%%%%%%%%%%%%%%%%%%%%%%%%%%%%%%%%%%%%%%%%%%%%%%%%%%
\begin{align}\label{cauchy-H} \left \{\begin{array}{cc}\widetilde{{\cal D}}_k \widetilde{u}(t, z)=\frac{\partial^2}{\partial
t^2}\widetilde{u}(t, z)
,(t, z)\in \R^\ast_+\times \H \\  \widetilde{u}(0, z)=0, \widetilde{u}_t(0, z)=\widetilde{u}_1(z), \widetilde{u}_1\in
C^\infty_0(\H)\end{array}
\right. \end{align}

%%%%%%%%%%%%%%%%%%%%%%%%%%%%%%%%%%%%%%%%%%%%%%%%%%%%%%%%%%
\begin{align}\label{cauchy-M} \left \{\begin{array}{cc}\Lambda^{\lambda, k}_X w(t, X)=\frac{\partial^2}{\partial
t^2}w(t, X)
,(t, X)\in \R^\ast_+\times \R\\ w(0, X)=0, w_t(0, X)=w_1(X), w_1\in
C^\infty_0(\R)\end{array}
\right. \end{align}

 The modified Schr\"odinger operators with magnetic field on the hyperbolic disc $\D$  is given by
\begin{align}{\cal D}_{k}= {\cal L}_{k}^{\D}+k^2+\frac{1}{4}\end{align} 
where ${\cal L}_{k}^{\D}$ is the  Schr\"odinger operator with constant magnetic field on the hyperbolic disc $\D$ given in \cite{FER-VESEL} by
 \begin{align}{\cal L}_{k}^{\D}=(1-|w|^2)^2\frac{\partial^2}{\partial w \partial \overline{w}}+k(1-|w|^2)w\frac{\partial}{\partial w}& \nonumber\\
-k(1-|w|^2)\overline{w}\frac{\partial}{\partial \overline{w}}-k^2|w|^2\end{align}

 The modified Schr\"odinger operators with magnetic field on the hyperbolic half plane $\H$  is given by
\begin{align}\widetilde{{\cal D}}_{k}=y^2\left(\frac{\partial^2}{\partial x^2}+\frac{\partial^2}{\partial y^2}\right)+ 2 i k y\frac{\partial}{\partial x}+ \frac{1}{4}\end{align}
which can be written as
\begin{align}\widetilde{{\cal D}}_k= {\cal L}_{k}^{\H}+k^2+\frac{1}{4}\end{align} 
where ${\cal L}_{k}^{\H}$ is the  Schr\"odinger operator with constant magnetic field on the hyperbolic upper half plane $\H$ given in \cite{IKEDA} by
\begin{align}{\cal L}_{k}^{\H}=-(z-\overline{z})^2)\frac{\partial^2}{\partial z\partial\overline{z}}+k(z-\overline{z})\left(\frac{\partial}{\partial z}+\frac{\partial
}{\partial \overline{z}}\right)-k^2\end{align}
the Schr\"odinger operator with the Morse potential is given for $X\in \R$ by \cite{IKEDA}
\begin{align}\Lambda^{\lambda, k}_X=\frac{\partial^2}{\partial X^2}-2k\lambda e^X -\lambda^2 e^{2X}\end{align}
or equivalently for $y \in \R^+$ by
\begin{align}\Lambda^{\lambda, k}_{\ln y}=\left(y\frac{\partial}{\partial y}\right)^{2}-2k\lambda y -\lambda^2 y^2\end{align}

The importance of the  Schr\"odinger operator with magnetic field and the Schr\"odinger operator with the Morse potential in both theory and application of mathematics and physics may be found in literature \cite{FAY}, \cite{IKEDA}, \cite{MORSE}, \cite{ TASS}. For example
the operators ${\cal L}_{k}^{\D}$(resp.  ${\cal L}_{k}^{\H}$)   has a physical interpretation as being the Hamiltonian which governs a non relativistic charged particle
moving under the influence of the magnetic field of constant strength $|k|$ perpendicular to $\D$ (resp.$\H$).
Also the purely vibrational levels of diatomic
molecules with angular momentum $l=0$ have been described by the Morse potential since $1929$ see \cite{MORSE}.\\
Note that for $\lambda, k \in \R$ the Magnetic Laplacians ${\cal D}_{k}$, $\widetilde{{\cal D}}_{k}$   and the Schr\"odinger operator with the diatomic molecular Morse potential $\Lambda^{\lambda, k}$ are,
non-positive, definite
and each of them has an absolute continuous spectrum as well as a points spectrum if $|k|\ge 1/2$ see \cite{B-I}.
For a recent work on the Morse Potential see \cite{A-B-M, FID, HASS-ZARE, ZNOJIL}
%%%%%%%%%%%%%%%%%%%%%%%%%%%%%%%%%%%%%%%%%%%%%%%%%%%%

%%%%%%%%%%%%%%%%%%%%%%%%%%%%%%%%%%%%%%%%%%%%%%%%%%%%%%%%%%

%%%%%%%%%%%%%%%%%%%%%%%%%%%%%%%%%%%%%%%%%%%%%%%%%%%%%%%%%%%%
\section{Explicit solutions for the wave equation with magnetic field on the hyperbolic disc}
This section is devoted to the linear wave equation associated to the modified Schr\"odinger operator with magnetic field  ${\cal D}_{k}$ on the hyperbolic disc  $\D$.\\ 
Let
  $\D=\{w\in \C, |w|<1\}$ be the unit disc endowed with the metric $ds$
\begin{align} ds^{2}=4
\frac{|dw|^2}{(1-|w|^2)^2}.\end{align}
Then the Rimannian manifold $(\D, ds)$ is the (conformal) Poincare disc model of the hyperbolic plane. The metric $ds$ is invariant with respect to the group \\
\begin{align} SU(1, 1)=\left\{ \left(
\begin{array}{cc}
A & \bar{B}\\
B & \bar{A}
\end{array}
\right) A, B\in \C: |A|^2-|B|^2=1\right\}.\end{align}
%%%%%%%%%%%%%%%%%%%%%%%%%%%%%%%%%%%%%%%%%%%%%%%%%%%%%%%%%%%%%%%%
 The hyperbolic surface form $d\mu(w)$ is given by
\begin{align}d\mu(w)=4(1-|w|^2)^{-2}dX dY.\end{align}
The hyperbolic distance $d(w, w')$  associated to $ ds $ is
\begin{align}\label{dist1}\cosh^2 (d(w, w')/2)=\frac{|1-w\overline{w'}|^2}{(1-|w|^2)(1-|w'|^2)}.\end{align}
The associated Laplace Beltrami operator is
 \begin{align}\Delta=\left(1-|w|^2\right)\frac{\partial^2}{\partial w\partial \bar{w}}.\end{align}

\begin{prop}\label{prop1}\cite{B-I}
For $k$ real number, we consider the projective representation $T^{k}$ of the group $G=SU(1, 1)$ on $\C^\infty(\D)$ defined by
\begin{align}T^{k}(g)f(z)=\left(\frac{\overline{Cw+D}}{Cw+D}\right)^{k}f(g^{-1
}w))\end{align}
where $g^{-1
}=\left(
\begin{array}{cc}
A & B\\
B & A
\end{array}\right)\in SU(1, 1).$\\
 i)The action $T^{k}$ is a unitary projective representation of the group $G$ on the Hilbert space $L^2(\D)=L^2(\D, d\mu)$.\\
ii) The Laplacian ${\cal D}_{k}$ is $T^{k}$ invariant, that is we have
\begin{align}\label{invar1}T^{k}(g){\cal D}_{k}={\cal D}_{k}T^{k}(g)\end{align}
 for every $g\in G$.\\
  Let \begin{align}\label{gw}g_w=\left(
\begin{array}{cc}
\frac{1}{\sqrt{1-|w|^2}} &
\frac{ w}{\sqrt{1-|w|^2}}\\
\frac{\bar{w}}{\sqrt{1-|w|^2}} & 
\frac{1}{\sqrt{1-|w|^2}}
\end{array}
\right)\end{align} then:\\
iii)
$g_w\in G$ and we have $g_w0=w$.\\
iv)
\begin{align}\label{Tkw}\left[T^{k}(g_{w'})f\right](w)=\left(\frac{1-\overline{w}w'}{1-w\overline{w'}}\right)^k f(g_{w'}^{-1
}w).\end{align}
 \end{prop}
\begin{lem} \label{lem1}
If    $u\in C^\infty (\mathbb{R}^+\times\D)$ be a radial function in the second variable and let  $\Phi \in C^\infty (\mathbb{R}^+\times\R^+)$  such that\\
  $u(t, w) =v(t, r)= \Phi(x, y)$ with  $x= \cosh^2(t/2)$, $y = \cosh^2(r/2)$ and $ r = d(0, w)$ then we have :  \\
i) ${\cal D}_{k} u(t, w) = l_{y}^k \Phi(x, y)$ with : \\
\begin{align}l^{k}_y \Phi(x, y) = \left[y(y-1)\frac{d^2}{dy^2}+(2y-1)\frac{d}{dy}
 +\frac{k^2}{y}+\frac{1}{4}\right]\Phi(x, y)\end{align}
ii) Setting $\Phi(x, y)=y^{-1/2}\Psi(x, y)$  we have\\
 $y^{1/2}l_{y}^ky^{-1/2}\Psi(x, y)=J_k \Psi(x, y)$ with
\begin{align}J_k \Psi(x, y)=[y(y-1)\frac{d^2}{dy^2}+y\frac{d}{dy}
 -\frac{1-4 k^2}{4 y}]\Psi(x, y)\end{align}
 iii) $\frac{d^2}{d t^2}u(t, w)=I_x \Phi(x, y)$ where\\
\begin{align} I_x\Phi(x, y)=[x(x-1)\frac{d^2}{dx^2}+(x-1/2)\frac{d}{dx}]\Phi(x, y)\end{align}

 \end{lem}
\begin{proof} Using the geodesic polar coordinates
$w=\tanh r/2\, \omega$, $r> 0$ and $\omega\in S^{1}$ we see that the radial part of the Magnetic Laplacian ${\cal D}_{k}$
 is given by
\begin{align}D_{k}=\frac{\partial^{2}}{\partial r^{2}}+\coth r \frac{\partial}{\partial r}+\frac{k^2 }{\cosh^2(r/2)}+\frac{1}{4}\end{align}
using the variables changes $y=\cosh^2(r/2)$, $x=\cosh^2(t/2)$  we get the result of i) and iii). The result ii) is simple and  is left to the reader.
\end{proof}
%%%%%%%%%%%%%%%%%%%%%%%%%%%%%%%%%%%%%%%%%%%%%%%%%%%%%%%%%%%%%%%%%%%%%%%%%
\begin{thm}\label{thm1} The wave equation in \eqref{cauchy-D} associated to the modified Schr\"odinger operators with magnetic field ${\cal D}_{k}$ on the hyperbolic disc $\D$
 has the solution
 \begin{align}\label{v}v_k(t,  r(0, w))=\frac{1}{2\pi}\left(\cosh^2(t/2) -\cosh^2(r/2))\right)_+^{-1/2}\nonumber\\ 
F\left(|k|, -|k|, \frac{1}{2}, 1-\frac{\cosh^2 (t/2)}{\cosh^2 (r/2)}\right)\end{align}
where  $F(a, b, c, z)$ is the Gauss hypergeometric function defined by:
\begin{align}
  F(a, b, c, z)=\sum_{n=0}^{\infty}\frac{(a)_n(b)_n}{(c)_n n!}z^n,
  \quad |z|<1,
\end{align}
where as usual $(a)_n$ is the Pochhamer symbol
  $(a)_n=\frac{\Gamma(a+n)}{\Gamma(a}$
and $\Gamma$ is the classical Euler function.
\end{thm}
\begin{proof} Using Lemma~\ref{lem1} the wave equation ${\cal D}_{k} v(t,  r(0, w)) =  \frac{d^2}{d t^2}v(t,  r(0, w))$
 is equivalent to $l_y^{k} \Phi(x, y)=I_x \Phi(x, y)$ which is equivalent to
$J^k_y \Psi(x, y)= I_x \Psi(x, y)$ with $\Phi(x, y)=y^{-1/2}\Psi(x, y)$. Setting $z=1-x/y$ we see that the last equation is equivalent to 
%\begin{align}\frac{\partial}{\partial y}=x^{-1}(1-z)^2\frac{\partial}{\partial z},
%\frac{\partial^2}{\partial y^2}=x^{-2}(1-z)^2\frac{\partial}{\partial z}(1-z)^2\frac{\partial}{\partial z}\end{align}
%\begin{align}\frac{\partial}{\partial x}=-y^{-1}\frac{\partial}{\partial z},
%\frac{\partial^2}{\partial x^2}=-y^{-2}\frac{\partial^2}{\partial z^2}\end{align}
\begin{align}\frac{1}{y}\left[z(1-z)\frac{d^2}{dz^2}+\left(3/2-2z\right)\frac{d}{dz}-\frac{1-4k^2}{4}\right]\varphi(z)=0\end{align}
and this is an hypergeometric equation with parameters $a=\frac{1}{2}-|k|, b=\frac{1}{2}+|k|, c=3/2$
 and an appropriate solution is see\cite{MAGN} p.42
\begin{align}\varphi_k(z)=\frac{1}{2\pi}z^{-1/2}F\left(|k|,-|k|, 1/2, z\right)\end{align}
that is 
\begin{align}v_k(t,  r(0, w))=y^{-1/2}\varphi_k(1-x/y) \end{align}
and the proof of the theorem ~\ref{thm1} is finished.\\
\end{proof}
\begin{lem}\label{lem2} If $u\in C^\infty_0(\D)$ and let $v_{k}(t, r(w, 0))$ be the function given in \eqref{v} then we have\\
i) $\lim_{t\longrightarrow 0}\int_{r(w, 0)<t}v_{k}(t, r(w, 0))u(w)d\mu(w)=0$\\
ii) $\lim_{t\longrightarrow 0}\frac{\partial}{\partial t}\int_{r(w, 0)<t}v_{k}(t, r(w, 0))u(w)d\mu(w)=u(0)$
\end{lem}
\begin{proof}  
Writing $w$ in the geodesic polar coordinates and making the change of variables
$x=\frac{\sinh^2 (r/2)}{\sinh^2 (t/2)}$ we have \\
$\int_{d(w, 0)<t} v_{k}(t, r(w, 0))u(w)d\mu(w)=\frac{\sinh (t/2)}{\pi}\int_0^1(1-x)^{-1/2}$
\begin{align}\label{u(t, w)}F\left(|k|, -|k|, 1/2, \frac{ (x-1)\sinh^2 (t/2)}{1+x\cosh^2(t/2)}\right)\nonumber\\ u^\#\left( 2 Arg\sinh(x^{1/2}\sinh (t/2))\right)dx\end{align}
where 
$u^\#(r)=\int_{S^1}u(\tanh r \omega)d\omega$
and it is not hard to see i) and ii) from  \eqref{u(t, w)}\\  
\end{proof}
\begin{thm}\label{thm2} 
Set $V_k(t, w, w')= T^ {k}(g_{w'})[v_{k}(t, r(0, w))]$ where $v_{k}(t, r(0, w))$ is given by
\eqref{v} and $T^ {k}(g_{w'})$ is as in \eqref{Tkw}then we have: 
\begin{align}\label{V}V_k(t, w, w')=\frac{1}{2\pi}\left(\frac{1-w\overline{w'}}{1-\overline{w}w'}\right)^{k} \left(\cosh^2 (t/2) -\cosh (d(w, w')/2)\right)_+^{-1/2}\nonumber\\
F\left(|k|, -|k|, \frac{1}{2}, 1-\frac{\cosh^2 (t/2)}{\cosh^2(d(w, w')/2)}\right)\end{align}
%%%%%%%%%%%%%%%%%%%%%%%%%%%%%%%%%%%%%%%%%%%%%%%%%%%%%%%%%%%%%%%%%%%%%%%%%%%%%%%%%%%%%%%%%%%%%%%%%%%
%%%%%%%%%%%%%%%%%%%%%%%%%%%%%%%%%%%%%%%%%%%%%%%%%%%%%%%%%%%%%%%%%%%%%%%%%%%%%%%%%%%%%%%%%%%
And we have
\begin{align}{\cal D}_{k}^wV_{k}(t, w,  w') = \frac{\partial^2}{\partial t^2}V_{k}(t, w,  w')\end{align}
\begin{align}{\cal D}_{-k}^{w'}V_{k}(t, w,  w')= \frac{\partial^2}{\partial t^2}V_{k}(t, w,  w')\end{align}
where ${\cal D}_{k}^w$ is the modified Schr\"odinger operators with magnetic field with respect to $w$
\end{thm}
\begin{proof}
The first formula is consequences of \eqref{v} and \eqref{Tkw}\\
For the last two results, Using the formula \eqref{invar1} we can write 
\begin{align}{\cal D}^{w}_{k}V_{k}(t, w,  w')& ={\cal D}^{w}_{k}T^ {k}(g_w')[v_{k}(t, r(0, w))]\end{align}
\begin{align}=T^ {k}(g_w'){\cal D}_{k}^{w}[v_{k}(t, r(0, w))]
=T^ {k}(g_w')\frac{\partial^2}{\partial t^2}[v_{k}(t, r(0, w))]\end{align}
\begin{align}=\frac{\partial^2}{\partial t^2}[T^ {k}(g_w')[v_{k}(t, r(0, w))]]
=\frac{\partial^2}{\partial t^2}V_{k}(t, w,  w')\end{align}
And
\begin{align}{\cal D}_{-k}^{w'}V_{k}(t, w,  w') ={\cal D}_{-k}^{w'}T^ {-k}(g_w)[v_{k}(t, r(0, w'))]\end{align}
\begin{align}=T^ {-k}(g_w){\cal D}_{-k}^{w'}[v_{k}(t, r(0, w'))]
=T^ {k}(g_w)\frac{\partial^2}{\partial t^2}[v_{k}(t, r(0, w'))]\end{align}
\begin{align}=\frac{\partial^2}{\partial t^2}[T^ {-k}(g_w)[v_{k}(t, r(0, w'))]]
=\frac{\partial^2}{\partial t^2}V_{k}(t, w,  w')\end{align}

\end{proof}
\begin{thm} \label{thm3}
The Cauchy problem \eqref{cauchy-D} for the wave equation associated to the modified Schr\"odinger operators with magnetic field on the hyperbolic disc has the the unique solution given by
\begin{align}\label{u} u(t, w)=\int_{d(w, w')<t}V_{k}(t, w, w')u_1(w')d\mu(w')\end{align}
with $V_{k}(t, w, w')$ is as in \eqref{V}

\end{thm}
%%%%%%%%%%%%%%%%%%%%%%%%%%%%%%%%%%%%%%%%%%%%%%%%%%%%%%%%%%%%%%%%%%%
\begin{proof} Using Theorem~\ref{thm2} it suffices to show that the function $u(t, z)$ defined by \eqref{u} satisfies the initial conditions
$u(0, w)=0$ and $u_t(0, w)=u_1(w)$, $ u_1\in C^\infty_0(\D)$. 
% For this we fix $w\in \D$ and write $w=g_w.0$ for some $g_w$ as in \eqref{gw}
\begin{align}u(t, w)=\int_{d(w, w')<t}T^ {-k}(g_{w})[v_{k}(t, r(0, w'))] u_1(w')d\mu(w')\end{align}
\begin{align}u(t, w)=\int_{d(w, w')<t}\left(\frac{1-w\overline{w'}}{1-w'\overline{w}}\right)^{-k} v_{k}(t, r(0, g_w^{-1} w')u_1(w')d\mu(w')\end{align}
Setting $g^{-1} w'=w''$ we obtain
\begin{align}u(t, w)=\int_{d(0, w'')<t}v_{k}(t, r(0, w''))\left(\frac{1-w\overline{g_w w''}}{1-\overline{w}g_w w''}\right)^{-k} u_1(g_ww'')d\mu(w'')\end{align}
Using Lemma~\ref{lem2} with $f(w'')=\left(\frac{1-w\overline{g_w w''}}{1-\overline{w}g_w w''}\right)^{-k} u_1(g_ww'')$
 we obtain the limit conditions and the proof of Theorem~\ref{thm3} is finished.\\
\end{proof}
%%%%%%%%%%%%%%%%%%%%%%%%%%%%%%%%%%%%%%%%%%%%%%%%%%%%%%%%%%%%%%%%
%%%%%%%%%%%%%%%%%%%%%%%%%%%%%%%%%%%%%%%%%%%%%%%%%%%%%%%%%%%%%%%
 Note that using the formula   \cite{MAGN},$ p.49$ \begin{align}
F(a, b, a+b+1/2, z)=F(2a, 2b, a+b+1/2, (1-\sqrt{1-z})/2)
\end{align} the wave kernel $V_{k}(t, w, w')$ can be written as
\begin{align}\label{wave-kernel}V_{k}(t, w, w')=\frac{1}{2\pi}\left(\frac{1-\overline{w}w'}{1-w\overline{w'}}\right)^k\left(\cosh^2(t/2)-
\cosh^2(d(w, w')/2)\right)_+^{-1/2}\nonumber\\ F\left(2|k|, -2|k|, 1/2, \frac{1}{2}\left(1-\frac{\cosh (t/2)}{\cosh (d(w, w')/2)}\right)\right)\end{align}
and this agrees with the formula obtained in \cite{I-M-2} with $\alpha=-\beta=k$

 %%%%%%%%%%%%%%%%%%%%%%%%%%%%%%%%%%%%%%%%%%%%%%%%%%%%%%%%%%%%%%%%%
%%%%%%%%%%%%%%%%%%%%%%%%%%%%%%%%%%%%%%%%%%%%%%%%%%%%%%%%%%%%%%%%%

%%%%%%%%%%%%%%%%%%%%%%%%%%%%%%%%%%%%%%%%%%%%%%%%%%%%%%%%%%%%%%%%%%%%%%%%%%%%%%%%%%%%%%%%%%%%%%%%%%%%%%%%%%%%%%%%%%%%%%%%%%%

%%%%%%%%%%%%%%%%%%%%%%%%%%%%%%%%%%%%%%%%%%%%%%%%%%%%%%%%%%%%%%%%%%
\section{Explicit solutions for the wave equation with magnetic field on the hyperbolic upper half plane}
In this section we give the solution of the Cauchy problem for the wave equation associated to the modified Schr\"odinger operator with magnetic field on the half plane model of the hyperbolic plane $\H$.\\
 It is well known that the Rimanian manifold $(\D, ds)$ has negative constant Gaussian curvature
and it is isometric via the Cayley transform:
\begin{align}w=c z=\frac{z-i}{z+i},  z=c^{-1}w=-i\frac{w+1}{w-1}\end{align}
 to the hyperbolic upper half plane: $\H=\{z=x+iy\in \C, y>0\}$ endowed with the usual hyperbolic metric
\begin{align}\widetilde{ds}^2=\frac{dx^2+dy^2}{y^2}.\end{align}
The metric $\widetilde{ds}$ is invariant with respect to the group $\widetilde{G}=SL_2(\R)$ with
 \begin{align}SL_2(\R)=\left\{ \left(
\begin{array}{cc}
a & b\\
c & d
\end{array}
\right) a, b, c, d \in \R: a d-c b=1\right\}.\end{align}

The hyperbolic surface form is
\begin{align}\widetilde{d\mu}(z)=\frac{1}{y^{2}}dx dy\end{align} and
 the hyperbolic distance $\rho(z, z')$ given respectively by
\begin{align}\label{dist2} \cosh^2 (\rho(z, z')/2)=\frac{(x-x')^{2}+(y+y')^{2}}{4yy'}.\end{align} 
The Laplace Beltrami operator
\begin{align}\widetilde{\Delta}=-(z-\overline{z})^2\frac{\partial^2}{\partial z\partial\overline{z}}.\end{align}

 \begin{prop} \label{prop3} i) For $k\in \R$
the Cayley transform induces the unitary operator $U_{k}$
from $L^2(\H)$ to $L^2(\D)$, $f \longrightarrow \left(U_{k}f\right)(w)$
 \begin{align}\label{op}\left(U_{k}f\right)(w)=\left(\frac{1-\overline{w}}{1-w}\right)^{k}f(c^{-1}w)\end{align}
where the Hilbert spaces $L^2(\H)$ and  $L^2(\D)$ are respectively
the space of complex-valued $\widetilde{d\mu}(z)$  respectively $d\mu (w)$ square integrable functions
on $\H$ respectively on $\D$.\\
ii) For $k\in \R$ the inverse of the operator $U_{k}$ is given by
\begin{align}\label{op-inverse}\left(U^{-1}_{k}g\right)(z)=\left(\frac{i-\overline{z}}{z+i}\right)^{k}g(c z)\end{align}
iii) For $k\in \R$ the following intertwining formula holds.
\begin{align}\label{intertwin}U_{k}\widetilde{{\cal D}}_{k}U^{-1}_{k}f(c^{-1}w)={\cal D}_{k} f(c^{-1}w)\end{align}
\end{prop}
The proof of this proposition is simple and in consequence is left to the reader.
%%%%%%%%%%%%%%%%%%%%%%%%%%%%%%%%%%%%%%%%%%%%%%%%%%%%%%%%%%%%%%%%%%%%%%%%%%%%%%%%%%%%%%%%%%%%%%%%%%%%%%%%5
%%%%%%%%%%%%%%%%%%%%%%%%%%%%%%%%%%%%%%%%%%%%%%%%%%%%%%%%%%%%%%%%%%%%%%%%%%%%%%%%%%%%%%%%%%%%%%%%%%%%%%%%%%%

%%%%%%%%%%%%%%%%%%%%%%%%%%%%%%%%%%%%%%%%%%%%%%%%%%%%%%%%%%%%%%%%%%%%%%%%%%%%%%%%%%%%%%%%%%%%%%%%%%%%%%%%%%%%
%%%%%%%%%%%%%%%%%%%%%%%%%%%%%%%%%%%%%%%%%%%%%%%%%%%%%%%%%%%%%%%%%%%%%%%%%%%%%%%%%%%%%%%%%%%%%%%%%%%%%%%%%%%%%%%

\begin{thm}\label{thm4} The Cauchy problem for the wave equation associated to the modified Schr\"odinger operator with magnetic field on the upper half plane model \eqref{cauchy-H} has the the unique solution given by
\begin{align}\widetilde{u}(t, z)=\int_{\rho(z, z')<t}\widetilde{V}_k(t, z, z')\widetilde{u}_1(z')dz'\end{align}
where $\widetilde{V}_k(t, z, z' )$ is given by
 \begin{align}\label{V-}\widetilde{V}_k(t, z, z')=\frac{1}{2\pi}\left(\frac{\overline{z}-z'}{\overline{z'}-z}\right)^k\left(\cosh^2 (t/2)-\cosh^2(\rho(z, z')/2)\right)_+^{-1/2}\nonumber\\
F\left(|k|, -|k|, \frac{1}{2}, 1-\frac{\cosh^2 (t/2)}{\cosh^2(\rho(z, z')/2)}\right)\end{align}
\end{thm}
\begin{proof}
  Using the formula \eqref{intertwin}  the problem \eqref{cauchy-H} is transformed into the problem \eqref{cauchy-D} 
  with $U_k(\widetilde{u}(t, z))(t, w)$ instead of $u(t, w)$
By \eqref{u} we have
\begin{align} U_k[\widetilde{u}(t, z)](w)=\int_{\D}\left[ V_k(t, w, w')\right](z)(U_k \widetilde{u_1})(w') d\mu(w')\end{align}
\begin{align} \widetilde{u}(t, z)=\int_{\D}U_k^{-1}\left[ V_k(t, w, w')\right](z)(U_k u_1)(w') d\mu(w')\end{align}
Using the formulas \eqref{op}  and \eqref{op-inverse} and by setting $w'=c z'$ we obtain 
\begin{align}\widetilde{u}(t, z)=\int_{\H} \widetilde{V}_k(t, z, z')\widetilde{u}_1(z') d\mu(z')\end{align}
with $\widetilde{V}_k(t, z, z')$ is given by \eqref{V-} and the proof of Theorem~\ref{thm4} is finished
\end{proof}
Note that using the formula 
$F(0,b;c,z)=1$ we obtain the solution of the wave equation in the hyperbolic plane \cite{A-B-M} see also
\cite{B-O-J},\cite{I-M-1}, \cite{I-M-2} and \cite{MO}.
\begin{cor}The Cauchy problem $(a)'',(b)''$ for the wave equation
on the hyperbolic plane has the unique solution given by
\begin{align}u(t,w)=\frac{1}{\sqrt{2}\pi}\int_{d(w,w')<t}(\cosh t-\cosh d(w,w'))^{-\frac{1}{2}}
f(w')d\mu(w')\end{align}
\end{cor}
\begin{prop} Let $V_{k}(t, z, z')$ be the kernels given in \eqref{V-} then we have
i)
\begin{align} \widetilde{V}_k(t, z, z')=
\frac{1}{2\pi}\left(\frac{\overline{z}-z'}{\overline{z'}-z}\right)^k\left(\cosh^2(t/2) -
\cosh^2(\rho(z, z')/2)\right)_+^{-1/2}\nonumber\\ \cos 2|k| \left(arc \cos \frac{\cosh (t/2)}{\cosh (\rho(z, z')/2)}\right)
\end{align} 
ii)
For $k$ integer or a half of an integer
\begin{align}\label{chebichev}\widetilde{V}_k(t, z, z')=\frac{1}{2\pi}\left(\frac{\overline{z}-z'}{\overline{z'}-z}\right)^k\left(\cosh^2(t/2)-
\cosh^2(\rho(z, z')/2))\right)_+^{-1/2}\nonumber\\ T_{2 |k|}\left(\frac{\cosh (t/2)}{\cosh (\rho(z, z')/2)}\right)\end{align}
where $T_{2k}(x)$ are the Chebichev polynomials of the first kind.\\
\end{prop}%%%%%%%%%%%%%%%%%%%%%%%%%%%%%%%%%%%%%%%%%%%%%%%%%%%
%%%%%%%%%%%%%%%%%%%%%%%%%%%%%%%%%%%%%%%%%%%%%%%% 

\begin{proof}

To see  i) we use \cite{MAGN} $p.39$ 
 $\cos \alpha z =F(\frac{\alpha}{2}, -\frac{\alpha}{2},  \frac{1}{2}, \sin^2 z)$ with
$z=arc \cos x$,  
and by \cite{MAGN}  $p.39$ 
 $T_n(1-2x)=F(-n, n,  \frac{1}{2}, x)$ 
we have ii).
\end{proof}
%%%%%%%%%%%%%%%%%%%%%%%%%%%%%%%%%%%%%%%%%%%%%%%%%%%
\section{Explicit solutions for the wave equation associated to the Schr\"odinger equation with the  Morse potential on $\R$}
The Schr\"odinger operator with the Morse potential is given for $X\in \R$ by
\begin{align}\Lambda^{\lambda, k}_X=\frac{\partial^2}{\partial X^2}-2k\lambda e^X -\lambda^2 e^{2X}\end{align}
or equivalently for $y \in \R^+$ by
\begin{align}\Lambda^{\lambda, k}_{\ln y}=\left(y\frac{\partial}{\partial y}\right)^{2}-2k\lambda y -\lambda^2 y^2\end{align}

\begin{prop}\label{prop5} 
i) The modified Schr\"odinger operator with magnetic field $\widetilde{{\cal D}}_{k}$ on the hyperbolic half plane  and  Schr\"odinger operator with the Morse diatomic molecular potential $\Lambda^{\lambda, k}_{\ln y}$ on the real line $\R$  are connected via the formulas
\begin{align}{\cal F}_x\left[y^{-1/2}{\cal D}^z_{k}y^{1/2}\Phi\right](\lambda, y)=\Lambda^{\lambda, k}_{\ln y}\left({\cal F}\Phi\right)(\lambda, y)\end{align}
where the Fourier transform is given by
\begin{align}
[{\cal F}f](\xi)=\frac{1}{\sqrt{2\pi}}\int_{\R}e^{-i x \xi}f(x)dx
\end{align}
ii) The wave kernels with the Morse potential $W_{\lambda, k}(t, y, y')$  is connected to the wave kernel with magnetic potential on the hyperbolic half plane $V_k(t, z, z')$ via the formula
\begin{align}\label{connection}W_{\lambda,  k}(t, y, y')= \frac{1}{\sqrt{y y'}}\int_{-\infty}^{\infty}e^{-i\lambda(x-x')}\widetilde{V}_k(t, z, z')d(x-x')\end{align}
\end{prop}
%%%%%%%%%%%%%%%%%%%%%%%%%%%%%%%%%%%%%%%%%%%%%%%%%%%%%%%%%%%%%%%%%%%%%%%%%%%%%%%%%%%%%%%%%%%%%%%%%%%%%%
%%%%%%%%%%%%%%%%%%%%%%%%%%%%%%%%%%%%%%%%%%%%%%%%%%%%%%%%%%%%%%%%%%%%%%%%%%%%%%%%%%%%%%%%%%%%%%%%%%%%

%%%%%%%%%%%%%%%%%%%%%%%%%%%%%%%%%%%%%%%%%%%%%%%%%%%%%%%%%%%%%%%%%%%%%%%%%%%%%%%%%%%%%%%%%%%%%%%%%%%%%
%%%%%%%%%%%%%%%%%%%%%%%%%%%%%%%%%%%%%%%%%%%%%%%%%%%%%%%%%%%%%%%%%%%%%%%%%%%%%%%%%%%%%%%%%%%%%%%%%%%%
\begin{proof} The part i) is simple and in consequance is left to the reader. To see ii)
by using i) the Cauchy problem for the wave equation with the Morse potential \eqref{cauchy-M} is transformed into
the Cauchy problem for the wave equation associated to the  modified Schr\"odinger operator with magnetic field on the hyperbolic half plane \eqref{cauchy-H} with $v(x, y)= y^{1/2} {\cal F}^{-1}[w(t, y, \lambda)](x)$
using the the formula \eqref{u} by virtue of injectivity of the Fourier transform we obtain\\
$y^{1/2} {\cal F}_\lambda^{-1}[w(t, y, \lambda)](x) =$
\begin{align} \int_0^\infty y'^{-1/2}\int_{\R}\widetilde{V}_k(t, z,   z') {\cal F}_{\lambda'}^{-1}[w_1(y', \lambda')](x') dy' dx'\end{align}
$y^{1/2} {\cal F}_\lambda^{-1}[w(t, y, \lambda)](x)=$
 \begin{align} \int_0^\infty y'^{-1/2}\left\{\widetilde{V}_k(t, z,   z')* {\cal F}^{-1}_{\lambda'}[w_1(y', \lambda')]\right\}(x)(\lambda) dy'\end{align}
\begin{align} w(\lambda, y, t)= \sqrt{2\pi}\int_0^\infty(y y')^{-1/2} {\cal F}_{x-x'}[\widetilde{V}_k(t, z,   z')](\lambda)w_1(y', \lambda) dy'\end{align}
that is\\
$w(\lambda, y, t)=$
\begin{align} \int_{0}^{\infty }\left[\frac{1}{ \sqrt{y y'}}\int_{-\infty}^{\infty}  e^{-i\lambda(x-x')}\widetilde{V}_k(t, z,   z') d(x-x')\right]   w_1(y', \lambda)  d y' \end{align}
which gives the formula \eqref{connection}.
\end{proof}
\begin{lem}\label{key-lem} Let
\begin{align}I_\lambda^{\alpha,\beta}(Y, Z)=\int_{-\infty}^\infty e^{-i\lambda x}(x+Y)^\alpha(Z^2-x^2)^\beta_+dx\end{align}
then we have\\
\begin{align}I_\lambda^{\alpha,\beta}(Y, Z)=\frac{[\Gamma(\beta+1)]^2}{\Gamma(2\beta+2)}(2Z)^{2\beta+1}(Y+Z)^\alpha e^{-i\lambda Y}\nonumber \\
\phi_1(\beta+1, -\alpha, 2\beta+2, 2i\lambda Y, \frac{2Y}{Y+Z})\end{align}
with the function $\Phi_1(a, b, c,  x, y)$ is the confluent hypergeometric function of two variables defined by the double series:(see for example \cite{ERD-TRIC} p. $225$).\\
\begin{align}\label{phi1}\Phi_1(a, b, c,  x, y)=\sum_{m, n \geq 0}\frac{(a)_{m+n}(b)_n}{c)_{m+n}m!n!}x^my^n\end{align}
 for $|y|<1$ and its analytic continuation elsewhere.
\end{lem}
\begin{proof} Set $x=Zs$ and $s=1-2\xi$ we get
\begin{align}I_\lambda^{\alpha,\beta}(Y, Z)=\frac{[\Gamma(\beta+1)]^2}{\Gamma(2\beta+2)}(2 Z)^{2\beta+1}(Y+ Z)^\alpha e^{-i\lambda Z}\nonumber \\
\int_0^1e^{2i\lambda Z\xi}\xi^\beta(1-\xi)^\beta\left(1-\frac{2 Z}{Y+ Z}\xi\right)^\alpha d\xi\end{align}
Using the fact that for $0< \Re \alpha <\Re\gamma$ the function $\Phi_1$ has the integral representation:\\
 $\Phi_1(a, b, c, x, y)=\frac{\Gamma(c)}{\Gamma(a)\Gamma(c-a)}\int_0^1u^{a-1}(1-u)^{c-a-1}(1-uy)^{-b}e^{ux}du$\\
 we see the result of Lemma \ref{key-lem}.\\
\end{proof}
%%%%%%%%%%%%%%%%%%%%%%%%%%%%%%%%%%%%%%%%%%%%%%%%%%%%%%%%%%%%%%%%%%%%
\begin{thm}\label{thm-sol-cauchy-M} For $k$ integer or half integer the Cauchy problem \eqref{cauchy-M} for the wave equation with
the Morse potential has the unique solution given by:\\
\begin{align}\label{sol-cauchy-M}w(t, y)=\int_{|\ln y-\ln y'|<t}W_{\lambda, k}(t, y, y')w_1(y')\frac{dy'}{y'}\end{align}
with
\begin{align}W_{\lambda, k}(t, y, y')=C_{k}(4 y y')^{-|k|}\left(\frac{\partial}{\sinh (t /2) \partial t}\right)^{2|k|}
\nonumber\\ \frac{(2Z)^{4|k|}}{(Z+Y)^{2|k|}} e^{-i\lambda Z}\Phi_1(2|k|+1/2, 2|k|, 4|k|+1, 2i\lambda Z, \frac{2 Z}{Z+Y})\end{align}
 $C_k=\frac{(-1)^{k}\Gamma(2|k|+1/2)}{2\Gamma(4|k|+1)\sqrt{\pi}}$,
 $Z=\sqrt{4y y'\cosh^2 (t/2)-(y+y')^2}$ \\ and
$Y=isign(k)(y+y')$
and as usual, the function $\Phi_1(a, b, c, x, y)$ is the confluent hypergeometric function of two variables defined by the double series \eqref{phi1}
\end{thm}
\begin{proof}From the formula \eqref{dist2} we can write
\begin{align}\label{dist-2}\left(\frac{\overline{z}-z'}{\overline{z'}-z}\right)^k\cosh^{-2|k|}(\rho(z, z')/2)=\nonumber \\ (-1)^{-k}(4y y')^{|k|}\left[x-x'+isign(k)(y+y')\right]^{-2|k|}\end{align}
From the Rodrigue formula for the Chebichev polynomials (\cite{MAGN} p.258)\\ 
$T_n(x)=a_n(1-x^2)^{1/2}\left(\frac{d}{dx}\right)^n[(1-x^2)^{n-1/2}]$
where $a_n={(-1)^n\sqrt{\pi}\over 2^{n+1}\Gamma(n+1/2)}$,
 and the formula \eqref{chebichev} we can write\\
\begin{align}\label{chebichev-rod} T_{2|k|}\left(\frac{\cosh (t/2)}{\cosh(\rho/2)}\right)=\frac{\sqrt{\pi}}{2\Gamma(2|k|+1/2)}
\left(\cosh^2(t/2)-\cosh^2(\rho/2) \right)^{1/2}\nonumber\\ \cosh^{-2|k|} (\rho/2) \left(\frac{\partial}{\sinh(t/2)\partial t}\right)^{2|k|}\left(\cosh^2(t/2)-\cosh^2(\rho/2) \right)^{2|k|-1/2} \end{align}
From the formulas \eqref{chebichev}, \eqref{dist-2} and \eqref{chebichev-rod}we have\\
\begin{align}
V_k(t, z, z')=c_k\left(\frac{\partial}{\sinh (t/2)\partial t }\right)^{2|k|}(4 y y')^{-|k|+1/2}\nonumber\\
[x-x'+isign(k)(y+y')]^{-2|k|}\nonumber\\
[4yy'\cosh^2 (t/2)-(y+y')^2-(x-x')^2]_+^{2|k|-1/2}\end{align}
where $c_k=\frac{(-1)^{-k}}{4\sqrt{\pi}\Gamma(2|k|+1/2)}$\\  Using ii) of Proposition \ref{prop5} we obtain\\
    $W_{\lambda, k}(t, y, y')=c_k(y y')^{-|k|}\left(\frac{\partial}{\sinh t/2  \partial t}\right)^{2|k|}
    \int^\infty_{-\infty}e^{-i\lambda x}[x+isign{k}(y+ y')]^{-2|k|}\times$\\  $[4 y y'\cosh^2 t-(y+y')^2-x^2]_+^{2|k|-1/2}dx$,
 \\ and by Lemma \ref{key-lem} we get the result Theorem \eqref{thm-sol-cauchy-M}
\end{proof}
%%%%%%%%%%%%%%%%%%%%%%%%%%%%%%%%%%%%%%%%%%%%%%%%%%%%%%%%%%%%%%%%%%%%%%
%%%%%%%%%%%%%%%%%%%%%%%%%%%%%%%%%%%%%%%%%%%%%%%%%%%%%%%%%%%%%%%%%

Note that using the formulas: 
$\phi_1(\alpha,0,\gamma,x,y)=F(\alpha,\gamma,x)$ and \cite{MAGN} p. 283
$F(\nu+1/2,2\nu+1,2iz)=\Gamma(1+\nu)e^{iz}(z/2)^{-\nu}J_{\nu}(z)$
we get by taking $k=0$ in\eqref{sol-cauchy-M} the solution of the wave equation with Morse(Liouville) potential \cite{A-B-M}

\begin{align}W_{\lambda, 0}=
\frac{1}{2}J_0(|\lambda|\sqrt{2e^{X+X'}(\cosh t-\cosh (X-X'))})
\end{align}
%%%%%%%%%%%%%%%%%%%%%%%%%%%%%%%%%%%%%%%%%%%%%%%%%%%%%%%%%%%%%%%%%%%%%%%%%%%%%%%
%%%%%%%%%%%%%%%%%%%%%%%%%%%%%%%%%%%%%%%%%%%%%%%%%%%%%%%%%%%%%%%%%%%%%%%%%%%%%%%%%%

%%%%%%%%%%%%%%%%%%%%%%%%%%%%%%%%%%%%%%%%%%%%%%%%%%%%%%%%%%%%%%%%%%%%%%%%%%%%%
%%%%%%%%%%%%%%%%%%%%%%%%%%%%%%%%%%%%%%%%%%%%%%%%%%%%%%%%%%%%%%%%%%%%%%

  M.V. Ould Moustapha, \textsc{Department of Mathematic,
 College of Arts and Sciences-Gurayat,
 Jouf University-Kingdom of Saudi Arabia }.\\
\textsc{Facult\'e des Sciences et Techniques
Universit\'e de  Nouakchott Al-asriya,
Nouakchott-Mauritanie.}\\
  \textit{E-mail address}: \texttt{mohamedvall.ouldmoustapha230@gmail.com}

\end{document}